\documentclass{amsproc}

\usepackage{etex}

\usepackage{verbatim}		

\usepackage{setspace} %\onehalfspacing
\usepackage{ragged2e}

\usepackage{cfr-lm}
\usepackage{stmaryrd}
\usepackage{cite}
\usepackage[all]{xy}
\usepackage[dvips]{graphicx}
\usepackage{graphpap}
\usepackage{ifthen}
\usepackage{enumerate}
\usepackage{amsmath}
\usepackage{mathrsfs}
\usepackage[errorshow]{tracefnt}
\usepackage{fancybox}
\usepackage{amsfonts}
\usepackage{amsmath}
\usepackage{amssymb}
\usepackage{amsthm}
\usepackage{multirow}
\usepackage{array}
\usepackage{mathtools}
\usepackage{soul}
\usepackage{pict2e}
\DeclareMathAlphabet{\mathpzc}{OT1}{pzc}{m}{it}

\DeclareFontFamily{U}{matha}{\hyphenchar\font45}
\DeclareFontShape{U}{matha}{m}{n}{
      <5> <6> <7> <8> <9> <10> gen * matha
      <10.95> matha10 <12> <14.4> <17.28> <20.74> <24.88> matha12
      }{}
\DeclareSymbolFont{matha}{U}{matha}{m}{n}

\DeclareMathSymbol{\pm}            {2}{matha}{"08}
\DeclareMathSymbol{\mp}            {2}{matha}{"09}
\DeclareMathSymbol{\varleftarrow}{3}{matha}{"D0}
\DeclareMathSymbol{\varrightarrow}{3}{matha}{"D1}
\DeclareMathSymbol{\vee}           {2}{matha}{"5F}
\DeclareMathSymbol{\wedge}         {2}{matha}{"5E}
\DeclareMathSymbol{\leq}         {3}{matha}{"A4}
\DeclareMathSymbol{\geq}         {3}{matha}{"A5}
\DeclareMathSymbol{\in}            {3}{matha}{"50}
\DeclareMathSymbol{\owns}          {3}{matha}{"51}

\usepackage{bbold}

\usepackage{tikz}

\makeatletter
\DeclareRobustCommand{\Lcorner}{\mathbin{\mspace{1mu}\text{\L@corner}\mspace{1mu}}}
\newcommand{\L@corner}{%
  \setlength{\unitlength}{\fontcharht\font`T}%
  \begin{picture}(0.8,0)
  \roundcap
  \Line(0.1,0.95)(0.1,0.05)
  \Line(0.1,0.05)(0.7,0.05)
  \end{picture}%
}
\makeatother

\makeatletter
\DeclareRobustCommand{\Tri}{\mathbin{\mspace{1mu}\text{\L@corneer}\mspace{1mu}}}
\newcommand{\L@corneer}{%
  \setlength{\unitlength}{\fontcharht\font`T}%
  \begin{picture}(0.8,0)
  \roundcap
  \Line(0.1,0.052)(1.0,0.052)
  \end{picture}%
}
\makeatother

\makeatletter
\DeclareRobustCommand{\Irt}{\mathbin{\mspace{1mu}\text{\L@corneeer}\mspace{1mu}}}
\newcommand{\L@corneeer}{%
  \setlength{\unitlength}{\fontcharht\font`T}%
  \begin{picture}(0.8,0)
  \roundcap
  \Line(0.1,0.82)(1.0,0.82)
  \end{picture}%
}
\makeatother

\makeatletter
\newcommand{\thickhline}{%
    \noalign {\ifnum 0=`}\fi \hrule height 1pt
    \futurelet \reserved@a \@xhline
}
\newcolumntype{'}{@{\hskip\tabcolsep\vrule width 1pt\hskip\tabcolsep}}
\makeatother
\newcolumntype{"}{@{\hskip\tabcolsep\vrule width 1.5pt\hskip\tabcolsep}}
\makeatother

\newcommand{\scr}{\mathscr}

\def\boxit#1{\vbox{\hrule\hbox{\vrule\kern3pt
             \vbox{\kern3pt#1\kern3pt}\kern3pt\vrule}\hrule}}

\newcommand{\beq}{\begin{equation}}
\newcommand{\beqn}{\begin{equation*}}
\newcommand{\eeq}{\end{equation}}
\newcommand{\eeqn}{\end{equation*}}
\newcommand{\beqa}{\begin{eqnarray}}
\newcommand{\beqan}{\begin{eqnarray*}}
\newcommand{\eeqa}{\end{eqnarray}}
\newcommand{\eeqan}{\end{eqnarray*}}
\newcommand{\bdm}{\begin{displaymath}}
\newcommand{\edm}{\end{displaymath}}

\newcommand{\ba}{\begin{array}}
\newcommand{\ea}{\end{array}}

\newcommand\nn{\nonumber}

\newcommand\benu{\begin{enumerate}}
\newcommand\eenu{\end{enumerate}}
\newcommand\bit{\begin{itemize}}
\newcommand\eit{\end{itemize}}

\def\dim{\mathrm{dim\,}}

\def\der'{\mathfrak{der}'\,}
\def\der{\mathfrak{der}\,}
\def\str'{\mathfrak{str}'\,}
\def\str{\mathfrak{str}\,}

\usetikzlibrary{positioning}

\newcommand{\de}{\delta}

\newcommand{\blb}%{\ensuremath
{\text{$\llbracket$\hspace{-4pt}\scalebox{0.99}{$|$}\hspace{-2.58pt}\scalebox{0.99}{$|$}\hspace{-2.58pt}\scalebox{0.99}{$|$}}}
\newcommand{\brb}%{\ensuremath{
{\text{$\rrbracket$\hspace{-4pt}\scalebox{0.99}{$|$}\hspace{-2.58pt}\scalebox{0.99}{$|$}\hspace{-2.58pt}\scalebox{0.99}{$|$}}}

\def\fg{{\mathfrak g}}

\def\*{\partial}

\newcommand{\cms}{{\sf a}}

\newtheorem{theorem}{Theorem}[section]

\theoremstyle{definition}

\theoremstyle{remark}

\numberwithin{equation}{section}

\begin{document}

\title{Non-associative structures in extended geometry}

%    Information for first author
\author{Martin Cederwall}
%    Address of record for the research reported here
\address{Department of Physics, Chalmers University of Technology, SE-412~96 Gothenburg, Sweden}
%    Current address
%\curraddr{Department of Mathematics and Statistics,
%Case Western Reserve University, Cleveland, Ohio 43403}
\email{martin.cederwall@chalmers.se}
%    \thanks will become a 1st page footnote.
%\thanks{The first author was supported in part by NSF Grant \#000000.}

%    Information for second author
\author{Jakob Palmkvist}
\address%{School of Science and Technology, \"Orebro University, SE-701~82 \"Orebro, Sweden}
%\curraddr
{Department of Mathematical Sciences, Chalmers University of Technology, SE-412~96 Gothenburg, Sweden}
%Case Western Reserve University, Cleveland, Ohio 43403}
\email{jakob.palmkvist@chalmers.se}
\thanks{This note is related to the talks given by the second author at the 
joint AMS-UMI meeting ``New developments in infinite-dimensional Lie algebras,
vertex operator algebras and the monster'' in Palermo in July 2024
(organised by Darlayne Addabbo, Lisa Carbone and Roberto Volpato)
and at 
the mini-workshop ``Infinite-dimensional Kac--Moody Lie algebras in supergravity and M-theory''
at Oberwolfach in November 2024 (organised by Guillaume Bossard, Axel Kleinschmidt and Hermann Nicolai).
The second author would like to thank
the organisers of both workshops
for the invitations 
as well as for helpful questions and comments during the talks.
He would also like to thank Per Bäck, Jonas T.\ Hartwig 
and Jonathan Nilsson for discussions.}

\date{} %January 1, 1994 and, in revised form, June 22, 1994.}

%\dedicatory{This paper is dedicated to our advisors.}


\begin{abstract}
We consider a generalisation of vector fields on a vector space, where the vector space is generalised to a highest-weight module over a Kac--Moody algebra. The generalised vector field is an element in a non-associative superalgebra defined by the module and the Kac--Moody algebra. Also the Lie derivative of a vector field with respect to another
vector field is generalised and expressed in a simple way in terms of this superalgebra. It reproduces the generalised Lie derivative in the general framework of extended geometry, which in special cases reduces to the one in exceptional field theory, unifying diffeomorphisms with gauge transformations in supergravity theories.
\end{abstract}

\maketitle

\section{Introduction}

\textit{Extended geometry} is an approach to M-theory in physics based on the idea of
``geometrising'' duality symmetries that appear in dimensional reduction of eleven-dimensional
supergravity (the low-energy limit of M-theory) to maximal supergravity in lower dimensions.
This requires
generalisations of
geometrical concepts such as
vector fields and Lie derivatives.
The role of $\mathfrak{sl}(n)$ and its $n$-dimensional highest-weight representation
are in these generalisations played by
any Kac--Moody algebra $\mathfrak{g}$ and any irreducible highest-weight representation over $\mathfrak{g}$
with a
dominant integral weight $\lambda$
as the highest weight
\cite{Cederwall:2017fjm}. 
In the case of \textit{exceptional geometry}
(where $\mathfrak{g}=\mathfrak{e}_r$ and $\lambda$ is the fundamental weight $\Lambda_1$ associated to the vertex 
$1$ in the Dynkin diagram of $\mathfrak{e}_r$, with numbering of vertices as in Figure~\ref{dynkinfig}),
the transformation of a generalised vielbein one-form under a generalised Lie derivative
unifies $r$ of the eleven diffeomorphisms in eleven-dimensional supergravity with
gauge transformations associated to the three-form gauge field \cite{Hull:2007zu,PiresPacheco:2008qik,Coimbra:2011ky,Berman:2012vc,Hohm:2013vpa}.

\begin{figure}[tb]
%\begin{center}
%\scalebox{1}{
%\begin{figure}
\begin{picture}(265,65)(70,-5)
\put(73,-10){${\scriptstyle{0}}$}
\put(113,-10){${\scriptstyle{1}}$}
\put(153,-10){${\scriptstyle{2}}$}
\put(202,-10){${\scriptstyle{r-4}}$}
\put(242,-10){${\scriptstyle{r-3}}$}
\put(282,-10){${\scriptstyle{r-2}}$}
\put(322,-10){${\scriptstyle{r-1}}$}
\put(260,45){${\scriptstyle{r}}$}
\thicklines
\multiput(210,10)(40,0){4}{\circle{10}}
\multiput(215,10)(40,0){3}{\line(1,0){30}}
\put(155,10){\circle{10}}
\put(115,10){\circle{10}}
\put(75,10){\circle{10}}
\put(120,10){\line(1,0){30}}
\put(80,10){\line(1,0){30}}
\put(75,10){\line(1,1){3.5}}
\put(75,10){\line(1,-1){3.5}}
\put(75,10){\line(-1,1){3.5}}
\put(75,10){\line(-1,-1){3.5}}
%\multiput(160,10)(10,0){5}{\line(1,0){5}}
\multiput(160,10)(35,0){2}{\line(1,0){10}}
\multiput(175,10)(10,0){2}{\line(1,0){5}}
\put(250,50){\circle{10}} \put(250,15){\line(0,1){30}}
\end{picture}
%\end{center}
\caption{Dynkin diagram of the Borcherds superalgebra $\mathscr B$ in the case
$(\mathfrak{g},\lambda)=(\mathfrak{e}_r,\Lambda_1)$, corresponding to
exceptional geoemtry. Removing the ``grey vertex'' ($\otimes$) yields the Dynkin diagram of $\mathfrak{e}_r$.}
\label{dynkinfig}
\end{figure}

In the present note we will give an expression for the Lie derivative in ordinary geometry,
in the case of polynomial vector fields on a vector space.
It has the form of a derived bracket, $\mathscr L_U V =[dU,V]$, where $d$ is an odd differential
on an associative superalgebra.
The question that we aim to answer in this note is whether
also the generalised Lie derivative can be written as some form of
a derived bracket $\mathscr L_U V =[dU,V]$, although it is known
that generalised vector fields in general do not form a Lie algebra.
We will see that this is possible, but the underlying superalgebra will no longer be associative.

\section{Lie derivatives as derived brackets in a Weyl superalgebra}
Let $A$ be the 
associative superalgebra over a field $\mathbb{K}$
with $2n$ even generators 
$\widetilde E_a,\widetilde F^a$ and $2n$ odd generators $E_a,F^a$ (where $a=1,\ldots,n$)
satisfying the commutation relations
\begin{align}
[\widetilde E_a,\widetilde F^b]=[E_a,F^b]&=\delta_a{}^b\,,
\end{align}
and otherwise commuting.
(We use the notation $[-,-]$ and terms like \textit{commute} with respect to the parity grading,
for example $[E_a,E_b]=E_aE_b+E_bE_a=0$.)
The subalgebra of $A$
generated by
$E_a,F^a$ is a Clifford algebra with split signature $(n,n)$,
as can be seen by changing basis to $E_a\pm F^a$ (and considering these elements as even),
and the even subalgebra
generated by
$\widetilde E_a,\widetilde F^a$ is
a Weyl algebra. Accordingly, we may call the full algebra $A$
a \textit{Weyl--Clifford algebra} or a \textit{Weyl superalgebra}.

A polynomial vector field on $\mathbb{K}^n$
can be considered as an element in $A$ of the form
$U=U^aE_a$ where $U^a \in \mathbb{K}[\widetilde F^1,\ldots,\widetilde F^n]\subset A$.
(Here and below, summation over repeated indices up and down is understood.)
Define an odd differential $d$ 
on $A$ by
$dE_a=\widetilde E_a$ and $d\widetilde F^a=F^a$. 
By this we mean that $d$ is an odd derivation and squares to zero, so $d\widetilde E_a=dF^a=0$.
The Lie derivative of a vector field $V$ parameterised by another vector field
$U$ can then be written
\begin{align}
L_U V= [dU,V]&=[dU^a E_a,V^bE_b]+[U^a dE_a,V^bE_b]\nonumber\\
&=[\partial_c U^a F^c E_a,V^bE_b]+[U^a dE_a,V^bE_b]\nonumber\\
&=\partial_c U^aV^b[F^c E_a,E_b]+U^a[\widetilde E_a,V^b]E_b\nonumber\\
&=\partial_c U^aV^b(-\delta_b{}^c E_a)+U^a\partial_aV^bE_b\nonumber\\
&=(-V^a\partial_aU^b+U^a\partial_aV^b)E_b\,.
\end{align}
It follows that
\begin{align}
L_U(L_V W)&=[dU,[dV,W]]=[[dU,dV],W]+[dV,[dU,W]]\nonumber\\
&=[d[dU,V]]+[dV,[dU,W]]=
L_{L_U V} W + 
L_V(L_U W) \label{cov}
\end{align}
and 
\begin{align}
L_U V = [dU,V]=d[U,V]+[U,dV]=-[dV,U]=- L_V U\,. \label{antisym}
\end{align}

The relations (\ref{cov}) and (\ref{antisym}) show 
that the Lie derivative
can be considered as a bracket in a Lie algebra, where the vector fields are even
elements, although they are odd in the original superalgebra $A$.
Brackets constructed in this way from an original bracket and a differential
are called \textit{derived brackets}
\cite{Kosmann-Schwarzbach}.

The even subspace of $A$ with basis elements
$F^a E_b$ closes under the commutator and form the Lie algebra
$\mathfrak{gl}(n)=\mathfrak{sl}(n)\oplus \mathbb{K}$. 
When $\mathbb{K}=\mathbb{C}$, the simple component $\mathfrak{sl}(n)$
is the finite Kac--Moody algebra
$\mathfrak{a}_{r}$, with $r=n-1$ (and when $\mathbb{K}=\mathbb{R}$, the split real form of $\mathfrak{a}_{r}$).
The generators $E_a,\widetilde E_a,F^a,\widetilde F^a$ of $A$ are basis elements of $n$-dimensional modules over $\mathfrak{a}_{r}$
with highest weight $\Lambda_1$ or lowest weight $-\Lambda_1$.
In the generalisation to extended geometry that we will consider next,
$\mathfrak{a}_{r}$ and $\Lambda_1$ are replaced with
any Kac--Moody algebra $\mathfrak{g}$ and any dominant integral weight $\lambda$ of $\mathfrak{g}$,
with corresponding irreducible highest- and lowest-weight modules $R(\lambda)$ and $\overline{R(\lambda)}$
\cite{Cederwall:2017fjm}.

\section{Generalisation to extended geometry}

\subsection{Preliminaries}

All algebras that we consider in this note are $\mathbb{Z}$-graded superalgebras over $\mathbb{K}$. We assume that
the $\mathbb{Z}$-grading and the $\mathbb{Z}/2\mathbb{Z}$-gradings are compatible with each other,
so that even and odd elements have even and odd $\mathbb{Z}$-degrees, respectively. We denote the
$\mathbb{Z}$-degree of a homogeneous element $x$ by $\deg{x}$, and for an arbitrary element $x$
we write
$x=\sum_{k\in\mathbb{Z}}x_k$,
where $\deg{x_k}=k$.

Let $\mathfrak g$ be a Kac--Moody algebra of rank $r$ with Cartan subalgebra $\mathfrak{h}$ (considering
$\mathfrak g$ as
a $\mathbb{Z}$-graded Lie superalgebra, we have $\deg{\,\mathfrak{g}}=0$) and let $\lambda$
be an arbitrary dominant integral weight of $\mathfrak{g}$. 
We denote the bracket in $\mathfrak{g}$ by $\llbracket-,-\rrbracket$,
reserving $[-,-]$ for the commutator with respect to a product
in an algebra $\scr C$ that we will construct from the pair $(\mathfrak{g},\lambda)$,
and further extend to an algebra $\scr A$.
We assume that the Cartan matrix of $\mathfrak g$ is
of co-rank at most $1$ and symmetrisable, so that, for some non-zero numbers $d_i$
(which we take to be fixed),
multiplication from the left with $\mathrm{diag}(d_i)$
yields a symmetric matrix ($i=1,\ldots,r$).

We extend the derived algebra $\mathfrak{g}'=\llbracket \mathfrak{g},\mathfrak{g}\rrbracket$
of $\mathfrak{g}$ to $\mathscr B_0=\mathfrak{g}'\oplus \mathbb{K}$
by adding a basis element $h_0$ such that
$\llbracket h_0,e_i\rrbracket=-d_i\lambda_ie_i$ and
$\llbracket h_0,f_i\rrbracket=d_i\lambda_i f_i$.
Then there is a non-zero linear combination 
$L$ of $h_0,h_1,\ldots,h_r$, unique up to normalisation,
with $\llbracket L,\mathfrak{g}\rrbracket=0$.

Let ${\mathscr B}_{\pm1}$ be modules over $\mathfrak{g}'$, where ${\mathscr B}_{-1}\simeq {R(\lambda)}$
with highest-weight vector  $q$, and ${\mathscr B}_{1}\simeq \overline{R(\lambda)}$
with lowest-weight vector $p$. We extend ${\mathscr B}_{\pm1}$ to modules over 
$\mathscr B_0$ by $h_0 \cdot p = h_0 \cdot q=0$. 
We can then fix the normalisation of $L$
such that $L\cdot p=p$ and $L\cdot q=-q$.

\subsection{Construction of the algebra $\mathscr C$}
Let $\mathscr C$ be the 
$\mathbb{Z}$-graded superalgebra generated by $\mathscr B_0,\mathscr B_{\pm1}$, where 
$\deg{\scr B_k}=k$ for $k=0,\pm1$, modulo the relations
\begin{align}
pq&=1+L-h_0\,,& [x_0,y_0]&=\llbracket x_0,y_0 \rrbracket\,,\nonumber\\
qp&=-L+h_0\,,  & [x_0,y_{\pm1}]&= x_0 \cdot y_{\pm1} \label{relations}
\end{align}
and
$(xy)z=x(yz)$
for all $x,y\in \mathscr C$ such that $x,y \in \bigoplus_{k=0}^\infty{\scr C}_{\pm k}$
or $y,z \in \bigoplus_{k=0}^\infty{\scr C}_{\pm k}$.
Thus $\mathscr C$ is obtained from the free (non-associative) algebra $\scr F$
generated by $\mathscr B_0,\mathscr B_{\pm1}$ by factoring out the ideal generated by
all such relations.
The definition does not guarantee that the ideal is proper, so \textit{a priori} we do not know
whether $\mathscr B_0,\mathscr B_{\pm1}$ are embedded as subspaces.
We will now show that this is indeed the case.

Let $\scr F_\pm$ be the free $\mathbb{Z}$-graded
Lie superalgebra generated by the odd subspace $\scr F_{\pm1}=\scr B_{\pm1}$.
It can naturally be extended to a 
$\mathbb{Z}$-graded
Lie superalgebra $\scr B_0 \oplus \scr F_{\pm}$ (direct sum of subspaces),
where $\scr B_0$ and $\scr F_{\pm}$ are subalgebras and the 
adjoint action of $\scr B_0$ on $\scr F_{\pm1}$
is the action of $\scr B_0$ that $\scr F_{\pm1}=\scr B_{\pm1}$ is equipped with as a
$\scr B_0$-module. This action is then extended from $\scr F_{\pm1}$ to the whole of
$\scr F_{\pm}$ by the Jacobi identity.
Let $\bigoplus_{k=0}^\infty\scr U_{\pm k} $ be the universal enveloping algebra of $\scr B_0 \oplus \scr F_{\pm}$,
where $\scr U_0$ is the universal enveloping algebra of $\scr B_0$. Set $\scr U=\scr U_-\oplus \scr U_0\oplus \scr U_+$.
This is a vector space but not an algebra, since no product is defined for pairs
of elements where one has negative degree and the other positive degree.

Let $\langle -,-\rangle: \scr B_{1}\otimes \scr B_{-1} \to \mathbb{K}$ and
$\llbracket -,-\rrbracket: \scr B_{1}\otimes \scr B_{-1} \to \scr B_0$ be bilinear maps recursively given by
\begin{align}
\langle p,q\rangle&=1\,, &\langle z_0\cdot x_1,y_{-1} \rangle + \langle x_1,z_0\cdot y_{-1}\rangle&=0\,,
 \nonumber\\ \llbracket p,q\rrbracket&=h_0\,,
&\llbracket z_0\cdot x_1,y_{-1} \rrbracket + \llbracket x_1,z_0\cdot y_{-1}\rrbracket
&=\llbracket z_0, \llbracket x_1,y_{-1}\rrbracket\rrbracket\,.
\end{align}

\begin{theorem}
There is an isomorphism between $\mathbb{Z}$-graded vector spaces
from $\scr C$ to $\scr U$. The restriction to $\scr C_0 \oplus \scr C_\pm$ is
an isomorphism between associative algebras.
Moreover, the products of $x_1\in\scr B_1 \subset \scr U_1$ and $y_{-1}\in\scr B_{-1} \subset \scr U_{-1}$
are given by
\begin{align}
x_1y_{-1}&=1+\langle  x_1, y_{-1} \rangle k-\llbracket x_1,y_{-1}\rrbracket\,,
\nonumber\\
y_{-1}x_1&=-\langle  x_1, y_{-1}\rangle k+\llbracket x_1,y_{-1}\rrbracket\,. \label{products}
\end{align}
\end{theorem} 
\begin{proof}
In refs. \cite{Cederwall:2023stz} (see also ref. \cite{Cederwall:2022oyb}),
it was shown that 
the associative product on $\scr U_0 \oplus \scr U_{\pm}$ can be extended to a product
defined on the whole of 
$\scr U_- \oplus \scr U_0 \oplus \scr U_+$
such that
(\ref{products}) and (\ref{relations}) hold,
as well as $(xy)z=x(yz)$
for all $x,y\in \mathscr U$ such that $x,y \in \bigoplus_{k=0}^\infty{\scr U}_{\pm k}$
or $y,z \in \bigoplus_{k=0}^\infty{\scr U}_{\pm k}$.
Thus there is a surjective algebra homomorphism $\psi$
from $\scr C$ to $\scr U$. In order to show that it is injective, we note that
$\bigoplus_{k=0}^\infty{\scr C}_{\pm k}$ is an enveloping algebra of $\scr B_0 \oplus \scr F_{\pm}$,
which means that the restriction of $\psi$ to $\bigoplus_{k=0}^\infty{\scr C}_{\pm k}$ is a surjective
homomorphism between enveloping algebras of $\scr B_0 \oplus \scr F_{\pm}$. It then follows from
the universal property of $\bigoplus_{k=0}^\infty{\scr U}_{\pm k}$ that this homomorphism is also injective.
\end{proof}

The algebra $\scr C$ is in general not associative, and 
does not even satisfy the weaker condition that commutator satisfies the Jacobi identity
(in other words, it is not \textit{Lie admissible}), so we do not get a Lie superalgebra
directly from it as its commutator algebra.
However, since the associative law $(xy)z=x(yz)$
holds whenever $x,y \in \bigoplus_{k=0}^\infty{\scr C}_{\pm k}$
or $y,z \in \bigoplus_{k=0}^\infty{\scr C}_{\pm k}$, the subspace
${\mathscr C}_{-1} \oplus {\mathscr C}_0 \oplus {\mathscr C}_{1}$ forms what is called a {\it local}
Lie superalgebra, and from it, we can get a Lie superalgebra \cite{kacSimpleIrreducibleGraded1968}.
Also the subspace $\scr B_{-1}\oplus \scr B_0 \oplus \scr B_1$ forms a local Lie superalgebra
with respect to the bracket $\llbracket -,-\rrbracket$.
This is the local part of a 
Borcherds superalgebra with a Cartan matrix obtained by 
adding a row and column $0$ to the Cartan matrix of $\fg$,
where the entries are given by
\begin{align}
\cms_{00}&=0\,, &\cms_{0i}&=-d_i\lambda_i\,, & \cms_{i0}&=-\lambda_i\,.
\end{align}
The corresponding Dynkin diagram in the case $(\mathfrak{g},\lambda)=(\mathfrak{e}_r,\Lambda_1)$
is given in Figure~\ref{dynkinfig} (with the ``grey'' vertex corresponding to row and column $0$).

Let $\bigoplus_{k=0}^\infty \scr T_{\pm k}$ be the tensor algebra of $\scr B_{\pm 1}$,
where $\scr T_0=\mathbb{K}$ and $\scr T_{\pm1}=\scr B_{\pm 1}$. It can be considered as the
universal enveloping algebra of $\scr F(\scr B_{\pm 1})$, and thus as a subalgebra of
$\bigoplus_{k=0}^\infty \scr U_{\pm k}$, which we identify with
$\bigoplus_{k=0}^\infty \scr C_{\pm k}$ by the above isomorphism $\psi$.
Then the vector space
$\scr T= \scr T_-\oplus \scr T \oplus \scr T_+$ is a subspace of $\scr C$. More precisely,
$\scr C_k = \scr U_0 \scr T_k = \scr T_k \scr U_0 $ for any $k\in\mathbb{Z}$.

\subsection{Extension to $\mathscr A$}

We extend $\mathscr C$ to a superalgebra $\mathscr A$ by adding two subspaces $\widetilde{\mathscr B}_{\pm1}$
with $\deg{\widetilde{\mathscr B}_{\pm1}}=0$,
which are $\mathfrak{g'}$-modules
isomorphic to $\scr B_{\pm1}$, but with opposite parity as subspaces of $\mathscr A$
(as such they are even, whereas
$\scr B_{\pm1}$ are odd).
Let $\varphi$ be a $\mathfrak{g}'$-module isomorphism:
\begin{align}
\varphi : \mathscr B_{\pm1} \to \widetilde{\mathscr B}_{\pm1}\,,\quad \varphi(y_{\pm1})=
\widetilde y_{\pm1}\,,\quad \varphi(x_0\cdot y_{\pm1})=x_0\cdot\varphi(y_{\pm1})=x_0\cdot\widetilde y_{\pm1}\,.
\end{align}
Thus $\widetilde {\mathscr B}_{-1}$ has a highest-weight vector $\widetilde q$
and $\widetilde{\mathscr B}_1$ has a lowest-weight vector $\widetilde p$.

Let $\mathscr A$ be the algebra generated by
$\mathscr C$ and
$\widetilde{\mathscr B}_{\pm1}$
modulo the relations
\begin{align}
[\mathscr C,\widetilde {\mathscr B}_{\pm1}]=[\widetilde {\mathscr B}_{\pm1},\widetilde {\mathscr B}_{\pm1}]&=0\,,&
[\widetilde x_1,\widetilde y_{-1}]&=\langle x,y\rangle\,.
\end{align}

It follows that the subalgebra
$\mathscr W$ of $\mathscr A$ generated by $\widetilde{\mathscr{B}}_{\pm1}$ is
an even Weyl algebra with $2\,\dim{R(\lambda)}$ generators,
and since it commutes with the subalgebra $\mathscr C$,
the whole algebra $\mathscr A$ is the tensor product of the subalgebras
$\mathscr W$ and $\mathscr C$.

\subsection{The map $d$} Consider the subspace $\mathscr W \mathscr T$ of $\scr A$. We define a linear map 
$d: \mathscr W \mathscr T
\to \mathscr A$
in three steps. First,
on $\mathscr B_{\pm1}$ and $\widetilde{\mathscr B}_{\pm1}$ it acts by
\begin{align}
d\widetilde x_{-1}&=x_{-1}\,, & dx_{1}&=\widetilde x_{1}\,, & d\widetilde x_1 = d x_{-1}&=0\,.
\end{align}
Second, on each of the subspaces $\mathscr W$ and 
$\mathscr T$,
it acts as an odd derivation.
Third, on a product $xy$, where $x\in\mathscr W$ and $y \in \mathscr T$, it acts by the Leibniz rule 
\begin{align}
d(xy)=(dx)y+x(dy)\,.
\end{align}

Let $\scr S$ be the even commutative subalgebra of $\scr W$ (and thus of $\scr A$) generated by 
$\widetilde {\scr B}_{-1}$.
We can now define a generalised vector field as an element in 
$\mathscr S\mathscr B_1 \subset \mathscr A$ and the generalised Lie derivative of a vector
field $V$ with respect to another one $U$ as $\mathscr L_U V =[dU,V]$.
This reproduces the known expressions in exceptional geometry, as we will see
in the next section.

\section{Explicit calculations in the case of finite $\mathfrak{g}$}\label{expl-section}

We will now compare the definition $\mathscr L_U V =[dU,V]$ of the generalised Lie derivative
with the established definition 
in the cases of extended geometry
where $\mathfrak{g}$ is finite-dimensional (which implies that the modules $\scr B_{\pm1}$ are finite-dimensional too).
The latter is motivated by the
requirement that it (in the case of exceptional geometry) give a
unification of
ordinary diffeomorphisms with gauge transformations
in supergravity \cite{Hull:2007zu,PiresPacheco:2008qik,Coimbra:2011ky,Berman:2012vc,Hohm:2013vpa}. In this definition,
the transformations of the components can be written \cite{Berman:2012vc}
\begin{align}
\mathscr L_U V^M &= L_U V^M + Y^{MN}{}_{PQ}\partial_N U^P V^Q\,\nonumber\\
&= U^N\partial_N V^M + Z^{MN}{}_{PQ}\partial_N U^P V^Q\, \label{genLieder}
\end{align}
with indices $M,N,\ldots=1,2,\ldots,\dim{\scr B_{\pm1}}$,
where $Y^{MN}{}_{PQ}$ and 
\begin{align}
Z^{MN}{}_{PQ}=Y^{MN}{}_{PQ}-\de_{P}{}^{M}\de_{Q}{}^N
\end{align}
are $\mathfrak{g}$-invariant tensors, which we now will give in an explicit
$\mathfrak{g}$-covariant form. 

Let $T_\alpha$, where $\alpha=1,\ldots,\dim{\mathfrak{g}}$, be basis elements of $\fg$.
Let $E_M$ and $F^N$ be basis elements of $\scr B_{1}$ and $\scr B_{-1}$, respectively,
such that $\langle E_M,F^N \rangle=\delta_M{}^N$
and let $t_\alpha{}_M{}^N$ be the entries of the corresponding representation matrices, meaning
\begin{align}
T_\alpha \cdot E_M &=-t_\alpha{}_M{}^N E_N\,, &
T_\alpha \cdot F^N &=t_\alpha{}_M{}^N F^M\,.
\end{align}

Let $\kappa$ be the invariant symmetric bilinear form on $\fg$
uniquely given by the condition $\kappa(e_i,f_j)=\de_{ij}/d_j$.
Then there is a vector space isomorphism $v:\mathfrak{h} \to \mathfrak{h}^\ast$
given by $v(h_i)(h_j)=\kappa(h_i,h_j)$ and a corresponding inner product on
$\mathfrak{h}^\ast$ given by $(\mu,\nu)=\big(v^{-1}(\mu),v^{-1}(\nu)\big)$.
Set $\kappa_{\alpha\beta}=\kappa(T_\alpha,T_\beta)$, and define $\kappa^{\alpha\beta}$ by
$\kappa_{\alpha\gamma}\kappa^{\gamma\beta}=\delta_\alpha{}^\beta$.

In all cases of exceptional geometry, we now have
\begin{align}
Z^{MN}{}_{PQ}&=-t_\alpha{}_Q{}^Mt_\beta{}_N{}^P \kappa^{\alpha\beta}
+\big((\lambda,\lambda)-1\big)\delta_Q{}^M \delta{}_N{}^P\,,
\end{align}
and it is then natural to take (\ref{genLieder}) with this choice of $Z^{MN}{}_{PQ}$
as the definition of the generalised Lie derivative in
general cases of extended geometry \cite{Berman:2012vc,Palmkvist:2015dea,Bossard:2017aae,Cederwall:2017fjm}.

The left pair of relations in (\ref{relations}) can now be covariantised to
\begin{align}
E_MF^N &= -t^\alpha{}_M{}^N T_\alpha - (\lambda,\lambda) \delta_M{}^N L +\delta_M{}^N L+ \delta_M{}^N\,,\nn\\
F^NE_M &= t^\alpha{}_M{}^N T_\alpha + (\lambda,\lambda) \delta_M{}^N L-\delta_M{}^N L\,.
\end{align}
This gives
\begin{align}
-[E_NF^M,F^P]=[F^ME_N,F^P]&=
-Z^{MP}{}_{QN}F^Q\,,\nn\\
-[E_NF^M,E_Q]=[F^ME_N,E_Q]&=
Z^{MP}{}_{QN}E_P\,\nn
\end{align}
and
\begin{align}
[dU,V]=[d(U^ME_M),V^NE_N]&=[dU^ME_M,V^NE_N]+[U^M dE_M,V^NE_N]\nn\\
&=[\partial_PU^M F^PE_M,V^NE_N]+[U^M \widetilde{E}_M,V^NE_N]\nn\\
&=\partial_PU^MV^N [F^PE_M,E_N]+U^M[\widetilde{E}_M, V^N] E_N\nn\\
&=\partial_PU^MV^N Z^{PQ}{}_{NM}E_Q+U^M\partial_MV^N E_N\nn\\
&=\partial_PU^MV^N Z^{QP}{}_{MN}E_Q+U^M\partial_MV^Q E_Q\nn\\
&=(U^M\partial_MV^Q+Z^{QP}{}_{MN}\partial_PU^MV^N)E_Q\,,
\end{align}
which shows that we indeed get back the generalised Lie derivative from the expression $[dU,V]$.

Let us 
compute the associators $(F^M,E_N,F^P)$ and $(E_M,F^N,E_P)$. 
We get
\begin{align}
(F^ME_N)F^P-F^M(E_NF^P)&=(F^ME_N)F^P+ F^MF^PE_N -\de_N{}^P F^M\nn\\
&=[F^ME_N,F^P]+F^PF^ME_N+ F^MF^PE_N -\de_N{}^P F^M\nn\\
&=(-Z^{MP}{}_{QN}-\de_N{}^P\de_Q{}^M)F^Q+2F^{(P}F^{M)}E_N \nn\\
&=-Y^{MP}{}_{QN}F^Q+2F^{(P}F^{M)}E_N
\end{align}
and
\begin{align}
(E_NF^M)E_Q-E_N(F^ME_Q)&=(E_NF^M)E_Q+ E_NE_QF^M -\de_Q{}^M E_N\nn\\
&=[E_NF^M,E_Q]+E_QE_NF^M+ E_NE_QF^M -\de_Q{}^M E_N\nn\\
&=(-Z^{MP}{}_{QN}-\de_N{}^P\de_Q{}^M)E_P+2E_{(Q}E_{N)}F^M \nn\\
&=-Y^{MP}{}_{QN}E_P+2E_{(Q}E_{N)}F^M\,,
\end{align}
or alternatively
\begin{align}
(F^ME_N)F^P-F^M(E_NF^P)&=
-F^M(E_NF^P)+\de_N{}^M F^P-E_NF^MF^P\nn\\
&=[E_NF^P,F^M]-E_NF^PF^M+\de_N{}^M F^P-E_NF^MF^P\nn\\
&=Z^{PM}{}_{QN}F^Q+\de_N{}^M \de_Q{}^P F^Q-2E_NF^{(P}F^{M)}\nn\\
&=Y^{PM}{}_{QN}F^Q-2E_NF^{(P}F^{M)}
\end{align}
and
\begin{align}
(E_NF^M)E_Q-E_N(F^ME_Q)&=-F^ME_NE_Q+\de_N{}^M E_Q-E_N(F^ME_Q)\nn\\
&=[F^ME_Q,E_N]-F^ME_QE_N-F^ME_NE_Q+\de_N{}^M E_Q\nn\\
&=(Z^{MP}{}_{NQ}+\de_N{}^M\de_Q{}^P)E_P-2F^ME_{(Q}E_{N)}\nn\\
&=Y^{MP}{}_{NQ}E_P-2F^ME_{(Q}E_{N)}\,
\end{align}
(where parentheses around indices denote symmetrisation).
This shows that the algebra $\scr A$ indeed is not associative.
In exceptional geometry, $(\fg,\lambda)=(\mathfrak{e}_r,\Lambda_1)$,
the $Y$-tensor is symmetric in its lower (and upper) indices up to $r=6$, so the antisymmetric part of
the associator still vanishes. However, for $r\geq7$, there is also an antisymmetric part of the
$Y$-tensor (which for $r=7$ factorises into a product of an invariant symplectic form and its inverse).

Only in the case $(\fg,\lambda)=(\mathfrak{a}_r,\Lambda_1)$, the term involving the
$Y$-tensor vanishes, but
still in this case, the second term remains. However, in this case the ideal
generated by $E_{(M}E_{N)}$ and $F^{(M}F^{N)}$ is proper and can be factored out, so that
the resulting algebra is the Weyl superalgebra $A$
that we started with.

\section{Outlook}

We end this note with some concluding remarks. We have restricted ourselves to the case of finite-dimensional
$\mathfrak{g}$ in Section~\ref{expl-section}, since the calculations there involve summations over all the 
$\dim{\,R(\lambda)}$ values of the indices (implicitly whenever indices are repeated up and down).
The expression $[dU,V]$ for the generalised Lie derivative $\mathscr{L}_U V$ makes sense also in the
infinite-dimensional cases.

Although we do not see 
any potential issues of the behaviour of the map $d$ in infinite-dimensional cases,
other issues will probably arise in further investigations.
In exceptional geometry, the relation $\mathscr{L}_U V=-\mathscr{L}_V U$ fails to hold already in the
last finite-dimensional case, $r=8$: the commutator of two generalised Lie derivatives gives not only another
generalised Lie derivative, but also an ``ancillary'' $\mathfrak{e}_r$-transformation.
In ref.~\cite{Cederwall:2017fjm}, we specified
in which cases of extended geometry ancillary transformations occur, and in ref.~\cite{Cederwall:2019qnw} we showed
how these occurrences can be derived from the structure of a corresponding \textit{tensor hierarchy algebra}
\cite{Palmkvist:2013vya,Bossard:2017wxl,Carbone:2018njd}.
In ref.~\cite{Cederwall:2022oyb,Cederwall:2023stz},
we showed how tensor hierarchy algebras can be constructed from the algebra $\mathscr{C}$
that we have considered here. However, this construction (called \textit{cartanification} in
ref.~\cite{Cederwall:2023stz}) only gives the desired result
in cases where ancillary transformation do not occur. The issues that arise beyond these cases where discussed
in the end of ref.~\cite{Cederwall:2023stz}, and it is an open question how this construction should be generalised.

\bibliographystyle{amsalpha}

%\bibliography{biblio}

%\end{document}

\end{document}